\documentclass[journal, draftclsnofoot, 11pt, a4paper, onecolumn]{IEEEtran}
\usepackage{graphicx}
\usepackage{amssymb}
\usepackage{amsmath}
\usepackage{calc}
\usepackage{array}
\usepackage{subfigure}
\usepackage{graphicx}
\usepackage{algorithm}
\usepackage{algorithmic}
\usepackage{amsthm}

\newtheorem{theo}{\textbf{Theorem}}

\newtheorem{prop}{\textbf{Proposition}}
\newtheorem{lem}{\textbf{Lemma}}

\theoremstyle{definition}
\newtheorem{defi}{\textbf{Definition}}
\newtheorem{ex}{Example}
\newtheorem{remk}{Remark}

\newcommand*{\Scale}[2][4]{\scalebox{#1}{$#2$}}

\begin{document}
\title{Topological Interference Management with \\ Reconfigurable Antennas}

\author{Heecheol Yang$^{*}$, Navid Naderializadeh$^{\dagger}$, A. Salman Avestimehr$^{\dagger}$, and Jungwoo Lee$^{*}$ \\
$^{*}$Dept. of Electrical and Computer Engineering, Seoul National University, Seoul, Korea \\
$^{\dagger}$Dept. of Electrical Engineering, University of Southern California \\ 
E-mail: hee2070@snu.ac.kr, naderial@usc.edu, avestimehr@ee.usc.edu, junglee@snu.ac.kr \vspace{-3mm}
}

\maketitle

\begin{abstract}
We study the symmetric degrees-of-freedom (DoF) of partially connected interference networks under linear coding strategies at transmitters without channel state information beyond topology. We assume that the receivers are equipped with reconfigurable antennas that can switch among their preset modes. In such a network setting, we characterize the class of network topologies in which half linear symmetric DoF is achievable. Moreover, we derive a general upper bound on the linear symmetric DoF for arbitrary network topologies. We also show that this upper bound is tight if the transmitters have at most two co-interferers.
\end{abstract}


\section{Introduction}
\label{sec_intro}
As wireless networks grow in size and get more complex, interference becomes a more challenging bottleneck to deal with. As a result, there has been a significant amount of research on developing communication strategies that exploit channel state information to efficiently manage the interference and increase network throughput. However, obtaining channel state information at transmitters (CSIT) is a burdensome task on communication systems. Thus, there has been a growing attention on interference management with a minimal effort to obtain CSIT.

Our focus in this paper is on the scenario in which interference management primarily relies on a coarse knowledge about channel states in the network, namely the ``topology'' of the network. Network topology simply refers to a 1-bit feedback information for each link between each transmitter and each receiver, indicating whether or not the signal of the transmitter is received above the noise floor at the corresponding receiver. Quite interestingly, it has been shown in \cite{TIM1} that the problem of interference management in this model (a.k.a., topological interference management) is equivalent to the index coding problem \cite{TIM2}-\cite{TIM3}, under the assumption that channel gains remain constant throughout the duration of communication, i.e., quasi-static fading channel. This model has also been studied for the fast fading scenario in \cite{TIM4}, where a structured repetition coding scheme was proposed to neutralize interference by only resending symbols at the transmitters in a structured way according to network topology. 
There has also been a general outer bound developed for this problem in \cite{TIM5}, by converting the problem to determining rank loss conditions for an ensemble of full-rank matrices with randomly-scaled rows.

In this paper, we consider the aforementioned topological interference management problem in a scenario that the receivers are equipped with reconfigurable antennas. Reconfigurabe antennas at the receivers have been considered to provide a diversity gain, offering a choice to receive the signal on its different preset modes, each of which has an independent channel gain. They provide additional opportunities for interference management, in particular alignment of interference via preset mode switching at the receivers. More specifically, with preset mode switching we can adjust the channel gains to change according to a desired pattern across time, hence increasing the flexibility for alignment of interference.

This problem of interference management with reconfigurable antennas at the receivers was first proposed in \cite{Gou} for $K$-user multiple-input single-output broadcast channels. It was extended to 3-user single-input single-output interference channels in \cite{IC2}. In particular, the author showed the condition for precoding matrices to align interference over time at the unintended receivers without knowledge about channel gain values at the transmitters. This approach was also applied to fully-connected $K$-user SISO interference channels in \cite{IC3}.

We present our results on the topological interference management problem in a scenario that the changing patterns of the channels can be designed as desired according to the network topology by reconfigurable antenna switching. Our contributions in this paper are three-fold. We first characterize the necessary and sufficient conditions for the topologies in which half linear symmetric degrees-of-freedom (DoF) is achievable. This corresponds to the ``best'' topologies since half symmetric DoF is the highest possible value for the symmetric DoF in the presence of interference. We next derive a general upper bound on the linear symmetric DoF for arbitrary network topologies.
We finally demonstrate that this upper bound is tight in network topologies in which each transmitter has at most two co-interferers, and we propose an optimal linear scheme for such network topologies.

We show that the benefits from network topology knowledge at the transmitters and preset mode switching at the receivers can be simultaneously exploited to achieve DoF gain for a large class of network topologies. We demonstrate that the desired changing patterns of the receiver preset modes provide an additional DoF gain, compared to the scenario where channel gains remain constant over time \cite{TIM1}. In particular, we highlight the distinction between the two scenarios by comparing the class of the best topologies and the upper bound on the symmetric DoF.

The rest of this paper is organized as follows.
In Section \ref{sec_model}, we introduce the system model. 
In Section \ref{sec_results}, we state our main results and their implications.
In Section \ref{sec_half}, we characterize the class of the ``best'' network topologies.
In Section \ref{sec_upper}, we derive the upper bound on the linear symmetric DoF for general network topologies.
In Section \ref{sec_2_coint}, we propose a new linear scheme for the class of the topologies with at most two co-interferers.
Finally, in Section \ref{sec_con}, we conclude this paper.

\emph{Notation}:
For a vector $\mathbf{a}$, 
$\textrm{diag}(\mathbf{a})$ represents a diagonal matrix whose diagonal entries are elements of $\mathbf{a}$.
For matrices $\mathbf{A}$ and $\mathbf{B}$,
$\textrm{dim}(\mathbf{A} \cap \mathbf{B})$ and $\textrm{dim}(\mathbf{A}, \mathbf{B})$ represent the dimensions of the intersection and the sum-space of $\textrm{span}(\mathbf{A})$ and $\textrm{span}(\mathbf{B})$, respectively.
For $a,b \in \mathbb{N}$, $[a:b]$ denotes $\{a, a+1, \ldots, b\}$. 

\section{System Model}
\label{sec_model}
Consider a $K$-user interference network consisting of $K$ transmitters ($\text{T}_{i}, i \in [1:K]$) that have a single transmit antenna each and $K$ receivers ($\text{R}_{j}, j \in [1:K]$) equipped with a reconfigurable antenna that can switch among its preset modes, each of which has an independent channel gain vector to all the transmitters. The preset mode of $\text{R}_{j}$ at time $t$ is denoted by $l_{j}(t)$.

We assume that channel gains remain constant throughout the course of communication, i.e. quasi-static fading channel.
In this scenario, since the channel gain at each time depends on the preset mode at the receiver, we represent the channel gain value from $\text{T}_{i}$ to $\text{R}_{j}$ at time $t$ as $h_{j,i}(l_{j}(t))\in \mathbb{C}$.
We also denote the channel matrix from $\text{T}_{i}$ to $\text{R}_{j}$ over $m$ channel uses as $$\mathbf{H}^{m}_{j,i}=\textrm{diag}\left([h_{j,i}(l_{j}(1)) ~ \ldots ~ h_{j,i}(l_{j}(m))]\right).$$

If the signal of $\text{T}_{i}$ is received at $\text{R}_{j}$ below the noise level, we ignore the link between them since we intend to characterize the network capacity at high SNR regime using the DoF metric. The connectivity pattern of the network, i.e., the set of links over which the transmitted signal is received above the noise level, is referred to as the \emph{network topology}. Therefore, the network graph is considered to be \emph{partially connected} in which some of the links are ignored.

We assume that each transmitter $\text{T}_{i}, i\in[1:K]$ intends to send a vector $\mathbf{s}_{i}=[s_{i,1} ~ \cdots ~ s_{i,n_{i}(m)}]^{T} \in \mathbb{C}^{n_{i}(m) \times 1}$ of $n_{i}(m)$ symbols to receiver $\text{R}_{i}$ over $m$ channel uses. In this paper, we restrict the transmission scheme to linear coding strategies at the transmitters, thus each transmitter $\text{T}_{i}$ sends
$\mathbf{x}^{m}_{i}=\mathbf{V}^{m}_{i}\mathbf{s}_{i}$ over $m$ channel uses,
where $\mathbf{V}^{m}_{i}=[\mathbf{v}^{m}_{i,1} ~\ldots ~\mathbf{v}^{m}_{i,n_{i}(m)}] \in \mathbb{C}^{m \times n_{i}(m)}$ is the beamforming matrix of $\text{T}_{i}$.
The received signal at $\text{R}_{j}$ over $m$ channel uses is given by
\begin{eqnarray}
\mathbf{y}^{m}_{j}=\mathbf{H}^{m}_{j,j}\mathbf{x}^{m}_{j}
+ \sum\nolimits_{i \in \mathcal{I}_j} \mathbf{H}^{m}_{j,i}\mathbf{x}^{m}_{i}+\mathbf{z}^{m}_{j},
\end{eqnarray}
where $\mathcal{I}_j$ is the set of transmitters interfering at $\text{R}_{j}$ and $\mathbf{z}^{m}_{j} \in \mathbb{C}^{m \times 1}$ is the additive white Gaussian noise vector over $m$ channel uses, whose entries are independent, each of which distributed as $\mathcal{CN}(0,1)$.

We assume that the transmitters have no knowledge about the channel values, but are only aware of the network topology, or equivalently $\{\mathcal{I}_j\}_{j=1}^{K}$. On the contrary, receivers know not only the network topology but also perfect channel state information. At the receivers, interfering signals are aligned through predetermined order of antenna switching, referred to as the \emph{preset mode pattern}, which depends on the network topology. 
We denote the preset mode pattern of $\text{R}_{j}$ during $m$ channel uses by $\mathbf{L}^{m}_{j}=[l_{j}(1) \ldots l_{j}(m)]$.
Since the channel value varies solely depending on the receiver's preset mode pattern, the channel values of the links towards the same receiver have the identical changing pattern.

Lastly, the linear DoF (LDoF) tuple $(d_{1},\ldots,d_{K})$ is said to be achievable if there exists a set of beamforming vectors and preset mode patterns for $j \in [1:K]$ almost surely, satisfying
\begin{eqnarray}
\label{D_tuple}
\begin{array}{c}
\hspace{-3mm} \textrm{dim}\left(\textrm{Proj}_{\mathcal{N}^{c}_{j}}\textrm{span}(\mathbf{H}^{m}_{j,j}\mathbf{V}^{m}_{j})\right) \geq d_{j}(m), \\ d_{j}=\underset{m\to\infty}{\textrm{lim}}\frac{d_{j}(m)}{m},
\end{array} \hspace{-5mm}
\end{eqnarray}
where $\textrm{Proj}_{\mathcal{A}^{c}}\mathcal{B}$ denotes the vector space induced by projecting $\mathcal{B}$ onto the orthogonal complement of $\mathcal{A}$ and $\mathcal{N}_{j}$ is the interference signal subspace at $\text{R}_{j}$ as 
\begin{eqnarray}
\mathcal{N}_{j}&=&\underset{i \in \mathcal{I}_{j}}{\bigcup} \textrm{span}(\mathbf{H}^{m}_{j,i}\mathbf{V}^{m}_{i}).
\end{eqnarray}

The linear symmetric DoF, denoted by $\textrm{LDoF}_{\textrm{sym}}$, is defined as the supremum $d$ for which the LDoF tuple $(d,\ldots,d)$ is achievable.

\section{Main Results}
\label{sec_results}
To present the main results of this paper, we first define the notions of the alignment and conflict graphs for a given network topology, which were first defined in \cite{TIM1}.

\begin{defi}
The \emph{alignment graph} is an undirected graph with vertex set $[1:K]$, where vertices $i$ and $j$ ($i\neq j$) are connected with a solid black edge if and only if $\text{T}_{i}$ and $\text{T}_{j}$ are connected to $\text{R}_{k}$ ($k \notin \{i, j\}$); i.e., $\{i, j\}\subseteq\mathcal{I}_{k}$.
\end{defi}

\begin{defi}
The \emph{conflict graph} is a directed graph with vertex set $[1:K]$, where vertex $i$ is connected to vertex $j$ ($i\neq j$) with a dotted red edge if and only if $\text{T}_{i}$ is connected to $\text{R}_{j}$; i.e., $i\in\mathcal{I}_{j}$.
\end{defi}

\noindent In contrast to the conflict graph defined in \cite{TIM1}, we assign a direction to the conflict edges. We also use the notions of alignment sets, internal conflicts, and conflict distance defined in \cite{TIM1}.

\begin{defi}
An \emph{alignment set} is defined as a set of vertices connected through alignment edges.
\end{defi}

\begin{defi}
An \emph{internal conflict} is defined as a conflict edge between two vertices that belong to the same alignment set.
\end{defi}

\begin{defi}
For two vertices that have an internal conflict between them, the \emph{conflict distance} is defined as the minimum number of alignment edges that are needed to be traversed to go from one vertex to the other.
\end{defi}

Finally, we define the notions of co-interferers and internal conflict cycles. 

\begin{defi}
For a transmitter $\text{T}_{i}$, its set of \emph{co-interferers}, denoted by $\widehat{\mathcal{T}}_{i}$, is defined as the set of transmitters interfering together with $\text{T}_{i}$ at one of its unintended receivers. To be precise,
\begin{eqnarray}
\widehat{\mathcal{T}}_{i} = \underset{i \in \mathcal{I}_{j}}{\bigcup} \mathcal{I}_{j} \setminus \{i\}.
\end{eqnarray}
\end{defi}

\begin{defi}
An \emph{internal conflict cycle} is defined as a directed cycle of conflict edges among nodes all of which receive interference from a single transmitter outside the cycle.
\end{defi}

Based on the above definitions, our main results are presented as follows.

\begin{theo}[Half linear symmetric DoF]\label{thm1}
For partially connected $K$-user interference networks with reconfigurable antenna switching between at least two modes at the receivers, half linear symmetric DoF can be achieved if and only if all vertices in the alignment and conflict graphs of the network have less than two incoming internal conflicts.
\end{theo}

\begin{remk}
The topologies in which half linear symmetric DoF is achievable correspond to the ``best'' topologies since half symmetric DoF is the highest possible value for the symmetric DoF in the presence of interference (i.e., as long as at least one interference link from $\text{T}_{i}$ to $\text{R}_{j}$ ($i\neq j$) exists in the network topology). Therefore, Theorem \ref{thm1} characterizes the best topologies for topological interference management with reconfigurable antennas at the receivers.
\end{remk}

\begin{remk}
Without reconfigurable antennas, it was shown in \cite{TIM1} that half linear symmetric DoF is achievable if and only if there is no internal conflict in the alignment and conflict graphs of the network, which is strictly more restrictive than the above condition in Theorem 1. Therefore, Theorem 1 demonstrates that the class of ``best'' network topologies is extended by utilizing preset mode switching of reconfigurable antennas at the receivers. In particular, if receivers are equipped with reconfigurable antennas, each receiver can distinguish its desired signal by preset mode switching although an interferer has the identical beamforming vector with its corresponding transmitter. Thus, if each vertex has at most one incoming internal conflict on the alignment and conflict graphs of the network, half linear symmetric DoF is achievable.
\end{remk}

\begin{theo}[Upper bound on the linear symmetric DoF for general topologies]\label{thm2}
For partially connected $K$-user interference networks with reconfigurable antenna switching among any number of preset modes at the receivers, the linear symmetric DoF is upper-bounded by
\begin{eqnarray}
\emph{LDoF}_{\emph{sym}} \leq \emph{min}\left(\frac{\Delta_{\emph{min}}+1}{2\Delta_{\emph{min}}+3},\frac{2L_{\emph{min,odd}}}{5L_{\emph{min,odd}}+1}\right),
\end{eqnarray}
where $\Delta_{\emph{min}}$ is the minimum conflict distance among the internal conflicts towards any vertex $j \in \mathcal{B}$, $\mathcal{B}$ is the set of vertices which have two or more incoming internal conflicts, and $L_{\emph{min,odd}}$ is the length of the shortest internal conflict cycle of odd length.
\end{theo}

\begin{remk}
If there are no vertices that have two or more incoming internal conflicts, i.e., $\mathcal{B}=\emptyset$, we can say that $\Delta_{\textrm{min}}$ is equal to infinity since there are no internal conflicts towards $\mathcal{B}$. In this case, the linear symmetric DoF upper bound of this network is equal to $\frac{1}{2}$, thus corresponding to the result of Theorem \ref{thm1}.
\end{remk}

\begin{remk}
Without reconfigurable antennas, it was shown in \cite[Corollary 8]{TIM1} that the symmetric DoF is upper-bounded by $\frac{\Delta}{2\Delta+1}$. Comparing this bound with the first bound in Theorem \ref{thm2}, we note the benefit of reconfigurable antenna switching, by increasing the upper bound from two perspectives. First, the minimum conflict distance is considered only among the internal conflicts towards each vertex $j \in \mathcal{B}$, while it was considered among all internal conflicts without mode switching. Second, the upper bound increases from $\frac{\Delta}{2\Delta+1}$ to $\frac{\Delta_{\textrm{min}}+1}{2(\Delta_{\textrm{min}}+1)+1}$ as if the minimum conflict distance increases from $\Delta$ to $\Delta_{\textrm{min}}+1$.
\end{remk}

\begin{remk}
If there exists an internal conflict cycle in the network topology, it is easy to verify that $\Delta_{\min}=1$, implying that the first upper bound in Theorem \ref{thm2} evaluates to $\frac{2}{5}$. In this case, if there exists at least one internal conflict cycle of \emph{odd} length, then the second bound will be strictly less than $\frac{2}{5}$, hence dominating the first bound.
\end{remk}

\begin{theo}[Linear symmetric DoF achievability for network topologies with at most two co-interferers]\label{thm3}
For partially connected $K$-user interference networks with reconfigurable antenna switching between at least two preset modes at the receivers, the linear symmetric DoF of $\frac{\Delta_{\textrm{min}}+1}{2\Delta_{\textrm{min}}+3}$ is achievable if the maximum number of co-interferers for each transmitter is bounded by two; i.e.,
\begin{eqnarray}
\underset{i\in [1:K]}{\emph{max}} \left|\widehat{\mathcal{T}}_{i}\right| \leq 2.
\end{eqnarray}
\end{theo}

\begin{remk}
Theorem \ref{thm3} shows the tightness of the upper bound in Theorem \ref{thm2} on the linear symmetric DoF for the class of network topologies in which all the transmitters have at most two co-interferers, or equivalently there is no fork in the alignment graph, where a fork is a vertex with three or more alignment edges connected to it. If some of the transmitters in the network have more than two co-interferers, the upper bound in Theorem \ref{thm2} may not be always tight.
\end{remk} 


\section{Identifying the ``Best'' Network Topologies}
\label{sec_half}
We dedicate this section to identify the ``best'' network topologies (i.e., the ones that can achieve half linear symmetric DoF), hence proving Theorem 1. We also give an example of such ``best'' network topologies.

We first state the two lemmas that will be used in the proof.

\begin{lem}[{\cite[Lemma 2]{IC2}}]
If $\mathbf{H}_{1}\mathbf{v}_{1}\in \emph{span}(\mathbf{H}_{2}\mathbf{V}_{2})$, where $\mathbf{H}_{1}$ and $\mathbf{H}_{2}$ are two $m \times m$ full-rank diagonal matrices whose diagonal entries have the same changing pattern, $\mathbf{v}_{1}$ is an $m \times 1$ column vector, $\mathbf{V}_{2}$ is an $m \times n$ thin matrix, i.e., $n<m$, 
and the entries of $\mathbf{v}_{1}$ and $\mathbf{V}_{2}$ are generated independently of the values of the diagonal entries of 
$\mathbf{H}_{1}$ and $\mathbf{H}_{2}$, then $\mathbf{v}_{1} \in \emph{span}(\mathbf{V}_{2})$.
\end{lem}

\begin{lem}[{\cite[Lemma 3]{Lash}}]
For two matrices $\mathbf{A}$ and $\mathbf{B}$ with the same row size,
\begin{eqnarray}
\emph{dim}\left(\emph{Proj}_{\mathcal{A}^{c}} \mathcal{B}\right) = \emph{rank} ([\mathbf{A} \hspace{1mm} \mathbf{B}]) - \emph{rank}(\mathbf{A}),
\end{eqnarray}
where $\mathcal{A}$ and $\mathcal{B}$ denotes $\emph{span}(\mathbf{A})$ and $\emph{span}(\mathbf{B})$, respectively.
\end{lem}

\begin{proof}[Proof of Theorem 1]
We first prove the converse. To prove the converse of Theorem 1, we show that half linear symmetric DoF cannot be achieved if some of vertices have two or more incoming internal conflicts. Suppose that each transmitter sends $\frac{m}{2}$ symbols over $m$ channel uses and a set of vertices $i$, $j$, and $k$ on the graph of the network are included in alignment set $\mathcal{S}$. If $\text{T}_{i}$ and $\text{T}_{j}$ are connected to $\text{R}_{k}$ and vertices $i$ and $j$ are connected with an alignment edge, the dimension of the desired signal subspace at $\text{R}_{k}$ not interfered by the interfering signals is given by
\begin{eqnarray*}
\textrm{dim}\left(\textrm{Proj}_{\mathcal{N}^{c}_{k}}\textrm{span}(\mathbf{H}^{m}_{k,k}\mathbf{V}^{m}_{k})\right) &=& \textrm{dim}(\mathbf{H}^{m}_{k,k}\mathbf{V}^{m}_{k}, \mathbf{H}^{m}_{k,i}\mathbf{V}^{m}_{i}, \mathbf{H}^{m}_{k,j}\mathbf{V}^{m}_{j}) - \textrm{dim}(\mathbf{H}^{m}_{k,i}\mathbf{V}^{m}_{i},\mathbf{H}^{m}_{k,j}\mathbf{V}^{m}_{j}) \\ 
&\leq& m - \textrm{dim}(\mathbf{H}^{m}_{k,i}\mathbf{V}^{m}_{i},\mathbf{H}^{m}_{k,j}\mathbf{V}^{m}_{j}).
\end{eqnarray*}

Due to (\ref{D_tuple}), the interfering signals should be aligned at $\text{R}_{k}$ to achieve half linear DoF as
\begin{eqnarray}
\label{eq:Rk}
\textrm{at $\text{R}_{k}$: } \textrm{dim}(\mathbf{H}^{m}_{k,i}\mathbf{V}^{m}_{i}, \mathbf{H}^{m}_{k,j}\mathbf{V}^{m}_{j}) \leq \frac{m}{2}.
\end{eqnarray}
Since the columns of each beamforming matrix are linearly independent, all $\mathbf{v}^{m}_{i,n}$ for $n \in [1:\frac{m}{2}]$ should be aligned at $\text{R}_{k}$ with $\mathbf{V}^{m}_{j}$ as
\begin{eqnarray}
\label{eq:viVj}
\textrm{at $\text{R}_{k}$: } \mathbf{H}^{m}_{k,i}\mathbf{v}^{m}_{i,n} \in \textrm{span}(\mathbf{H}^{m}_{k,j}\mathbf{V}^{m}_{j}).
\end{eqnarray}
Recall that the transmitters have no knowledge about channel realizations and $\mathbf{H}^{m}_{k,i}$ and $\mathbf{H}^{m}_{k,j}$ have the identical changing pattern of their diagonal entries. By Lemma 1,
\begin{eqnarray}
\label{eq:Vij1}
\mathbf{v}^{m}_{i,n} \in \textrm{span}(\mathbf{V}^{m}_{j}).
\end{eqnarray}
Since all the columns of $\mathbf{V}^{m}_{i}$ are included in $\mathbf{V}^{m}_{j}$,
\begin{eqnarray}
\textrm{dim}(\mathbf{V}^{m}_{i},\mathbf{V}^{m}_{j}) \leq \frac{m}{2},
\end{eqnarray}
and since each beamforming matrix has $m/2$ linearly independent columns, the above inequality implies that
\begin{eqnarray}
\textrm{span}(\mathbf{V}^{m}_{i})=\textrm{span}(\mathbf{V}^{m}_{j}).
\end{eqnarray}
Note that vertex $k$ is also in the same alignment set. Hence, following similar arguments as the above, we will have
\begin{eqnarray}
\label{eq:Vall2}
\textrm{span}(\mathbf{V}^{m}_{i})=\textrm{span}(\mathbf{V}^{m}_{j})=\textrm{span}(\mathbf{V}^{m}_{k}).
\end{eqnarray}
Meanwhile, equation (\ref{eq:viVj}) can be paraphrased by multiplying $\mathbf{H}^{m}_{k,k}$ and its inverse as
\begin{eqnarray}
\mathbf{H}^{m}_{k,i}\mathbf{v}^{m}_{i,n} \in \textrm{span}(\mathbf{H}^{m}_{k,k}(\mathbf{H}^{m}_{k,k})^{-1}\mathbf{H}^{m}_{k,j}\mathbf{V}^{m}_{j}).
\end{eqnarray}
By Lemma 1, it leads to
\begin{eqnarray}
\label{eq:Vij}
\mathbf{v}^{m}_{i,n} \in \textrm{span}( (\mathbf{H}^{m}_{k,k})^{-1}\mathbf{H}^{m}_{k,j}\mathbf{V}^{m}_{j}).
\end{eqnarray}
From (\ref{eq:Vall2}) and (\ref{eq:Vij}), $\textrm{span}(\mathbf{V}^{m}_{k})$ and $\textrm{span}( (\mathbf{H}^{m}_{k,k})^{-1}\mathbf{H}^{m}_{k,j}\mathbf{V}^{m}_{j})$ have non-zero intersection since $\mathbf{v}^{m}_{i,n}$ is included in both of the two subspaces. Thus, it can be expressed as
\begin{eqnarray*}
\textrm{dim}(\mathbf{H}^{m}_{k,k}\mathbf{V}^{m}_{k} \cap \mathbf{H}^{m}_{k,j}\mathbf{V}^{m}_{j}) > 0.
\end{eqnarray*}

At $\text{R}_{k}$, the desired signal subspace is contaminated by the interfering signal subspace from $\text{T}_{j}$. Thus, $\text{R}_{k}$ cannot decode all the desired $m/2$ symbols over $m$ channel uses, implying that half linear symmetric DoF cannot be achieved. This completes the proof of the converse.

We now prove the achievability. To prove the achievability of Theorem 1, we show that half linear symmetric DoF can be achieved if all of the vertices have less than two incoming internal conflicts. Suppose that each transmitter sends one data symbol in two time slots and the set of transmitters whose corresponding vertices are included in the same alignment set have the same $2 \times 1$ beamforming vector. We specify the beamforming vectors for each alignment set such that any two of them are linearly independent. Consider a receiver $\text{R}_{j}$ in the network. If $j$ is included in alignment set $\mathcal{S}_{1}$, $\text{R}_{j}$ receives the interfering signals from the transmitters included in another alignment set ($\mathcal{S}_{2}$) or from a single transmitter in alignment set $\mathcal{S}_{1}$. For the former case, interfering signals can be aligned to a 1-dimensional signal subspace if $\text{R}_{j}$ does not change the preset mode over two channel uses; i.e., $\mathbf{L}^{2}_{j}=[1 \,\, 1]$, since all its interferers have the same beamforming vector.
In this case, the received signal at $\text{R}_{j}$ over two channel uses is given by
\begin{eqnarray*}
\mathbf{y}^{2}_{j} =
\left[\begin{array}{cc} h_{j,j}(l_{j}(1)) & 0 \\ 0 & h_{j,j}(l_{j}(1))\end{array}\right] \mathbf{v}^{2}_{\mathcal{S}_{1}}s_{j,1}
+ \underbrace{\sum\limits_{i \in \mathcal{S}_2}\left[\begin{array}{cc} h_{j,i}(l_{j}(1)) & 0 \\ 0 & h_{j,i}(l_{j}(1))\end{array}\right] \mathbf{v}^{2}_{\mathcal{S}_{2}}s_{i,1}}_{\textrm{rank}=1} +\mathbf{z}^{2}_{j},
\end{eqnarray*}
where $\mathbf{v}^{2}_{\mathcal{S}_{1}}$ and $\mathbf{v}^{2}_{\mathcal{S}_{2}}$ are $2 \times 1$ beamforming vectors for alignment sets $\mathcal{S}_{1}$ and $\mathcal{S}_{2}$, respectively.
For the latter case, when $\text{R}_{j}$ is only interfered by a transmitter which is included in the same alignment set, it can have two independent equations by changing the preset mode from $1$ to $2$; i.e., $\mathbf{L}^{2}_{j}=[1 \,\, 2]$, although $\text{T}_{j}$ and the interferer have the same beamforming vector. It can be expressed as
\begin{eqnarray*}
\mathbf{y}^{2}_{j} =
\left[\begin{array}{cc} h_{j,j}(l_{j}(1)) & 0 \\ 0 & h_{j,j}(l_{j}(2))\end{array}\right] \mathbf{v}^{2}_{\mathcal{S}_{1}}s_{j,1}
+ \left[\begin{array}{cc} h_{j,i}(l_{j}(1)) & 0 \\ 0 & h_{j,i}(l_{j}(2))\end{array}\right] \mathbf{v}^{2}_{\mathcal{S}_{1}}s_{i,1} +\mathbf{z}^{2}_{j}.
\end{eqnarray*}

Thus, all users can achieve $1/2$ linear DoF and the proof is complete.
\end{proof}

\begin{ex}

\begin{figure}[t]
    \centerline{\includegraphics[width=8.0cm]{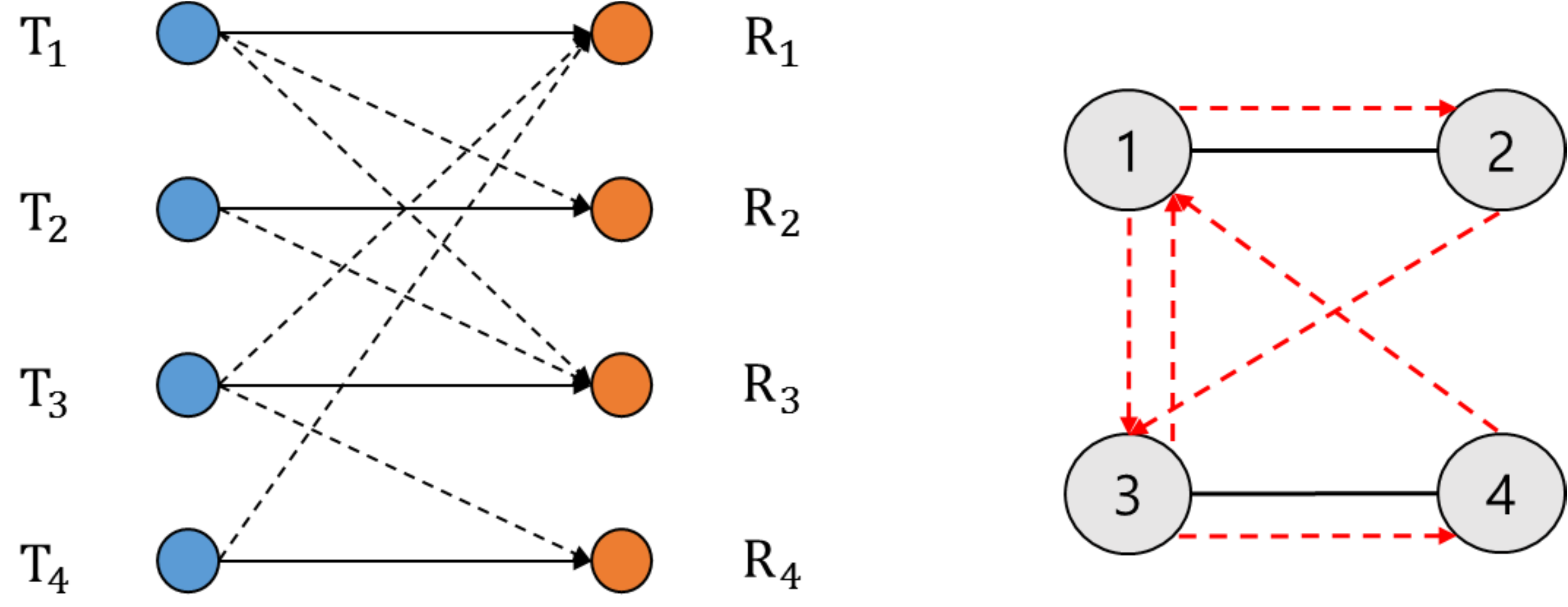}}
    \caption{A $4$-user interference network and its alignment and conflict graphs.}
    \label{DoF12}
    \vspace{-4mm}
\end{figure}

Consider a 4-user interference network illustrated in Fig. \ref{DoF12}. According to its corresponding alignment and conflict graph, the alignment sets are determined as
\begin{eqnarray*}
\mathcal{S}_{1} = \{1, 2\}, \,\,\,\,\, \mathcal{S}_{2} = \{3, 4\}.
\end{eqnarray*}

Without preset mode switching, this interference network cannot achieve half linear symmetric DoF since there exist internal conflicts (from vertex 1 to 2 and from vertex 3 to 4) on the graph. However, it can achieve half linear symmetric DoF with the aid of preset mode switching because no vertex has two or more incoming internal conflicts. We now show an achievable scheme for this interference network. Each user sends one data symbol to its corresponding receiver during two time slots. We determine the beamforming vectors for each alignment set in order to be linearly independent as
\begin{eqnarray*}
\mathbf{v}^{2}_{1}=\mathbf{v}^{2}_{2}=[1 \hspace{3mm} 1]^{T},
\hspace{5mm}\mathbf{v}^{2}_{3}=\mathbf{v}^{2}_{4}=[1 \hspace{3mm} -1]^{T},
\end{eqnarray*}
and the preset mode patterns are given by
\begin{eqnarray*}
\mathbf{L}^{2}_{1}=\mathbf{L}^{2}_{3}=[1 \hspace{3mm} 1],
\hspace{5mm}\mathbf{L}^{2}_{2}=\mathbf{L}^{2}_{4}=[1 \hspace{3mm} 2].
\end{eqnarray*}

Receivers $\text{R}_{1}$ and $\text{R}_{3}$ which are interfered by the transmitters in the other alignment set do not change their preset mode to align interfering signals into a 1-dimensional signal subspace, while $\text{R}_{2}$ and $\text{R}_{4}$ which are interfered by the transmitter in the same alignment set change their preset modes to have independent desired signal and interference subspaces. Hence, receivers $\text{R}_{1}$ and $\text{R}_{3}$ can decode their desired symbols by adding the received signals in time slots 1 and 2 and by subtracting the received signal in time slot 2 from time slot 1, respectively. Receivers $\text{R}_{2}$ and $\text{R}_{4}$ can also decode their desired symbols since they each have two independent equations during two time slots.

\end{ex}

\section{Upper Bound on the Linear Symmetric DoF for General Topologies}
\label{sec_upper}
In this section, we prove Theorem \ref{thm2} by providing two upper bounds on the linear symmetric DoF for arbitrary network topologies in terms of the minimum conflict distance for internal conflicts towards the vertices that have two or more internal conflicts and the length of the shortest internal conflict cycle of odd length. We prove the theorem through the following two propositions.

\begin{prop}\label{prop1}
For partially connected $K$-user interference networks with reconfigurable antenna switching among any number of preset modes at the receivers, the linear symmetric DoF is upper-bounded by
\begin{eqnarray}
\emph{LDoF}_{\emph{sym}} \leq \frac{\Delta_{\emph{min}}+1}{2\Delta_{\emph{min}}+3}.
\end{eqnarray}
\end{prop}
\begin{proof}
Suppose that each user sends $\lambda m$ data symbols over $m$ channel uses; i.e., $n_i = \lambda m, \forall i \in [1:K]$. For two arbitrary transmitters $\text{T}_{p}$ and $\text{T}_{q}$, if $\text{T}_{p}$ and $\text{T}_{q}$ are both interfering at $\text{R}_{s}$ and vertices $p$ and $q$ are connected with an alignment edge, the dimension of the desired signal subspace at $\text{R}_{s}$ not interfered by the interfering signals is given by
\begin{eqnarray*}
\textrm{dim}\left(\textrm{Proj}_{\mathcal{N}^{c}_{s}}\textrm{span}(\mathbf{H}^{m}_{s,s}\mathbf{V}^{m}_{s})\right) &=& \textrm{dim}(\mathbf{H}^{m}_{s,s}\mathbf{V}^{m}_{s}, \mathbf{H}^{m}_{s,p}\mathbf{V}^{m}_{p}, \mathbf{H}^{m}_{s,q}\mathbf{V}^{m}_{q}) - \textrm{dim}(\mathbf{H}^{m}_{s,p}\mathbf{V}^{m}_{p},\mathbf{H}^{m}_{s,q}\mathbf{V}^{m}_{q}) \\ 
&\leq& m - \textrm{dim}(\mathbf{H}^{m}_{s,p}\mathbf{V}^{m}_{p},\mathbf{H}^{m}_{s,q}\mathbf{V}^{m}_{q}).
\end{eqnarray*}

The interfering signals should be aligned at $\text{R}_{s}$ to set aside $\lambda m$-dimensional desired signal subspace as
\begin{eqnarray*}
\textrm{at $\text{R}_{s}$: } \textrm{dim}(\mathbf{H}^{m}_{s,p}\mathbf{V}^{m}_{p}, \mathbf{H}^{m}_{s,q}\mathbf{V}^{m}_{q}) \leq (1-\lambda)m.
\end{eqnarray*}
Since the columns of each beamforming matrix are linearly independent, we will have
\begin{eqnarray}
\label{eq:HVHV}
\textrm{dim}(\mathbf{H}^{m}_{s,p}\mathbf{V}^{m}_{p} \cap \mathbf{H}^{m}_{s,q}\mathbf{V}^{m}_{q}) &\geq& 2\lambda m - (1-\lambda)m, \\ \nonumber &=& (3\lambda -1)m.
\end{eqnarray}

Suppose that an $m \times 1$ vector $\mathbf{v}^{m}_{*}$ is included in the vector space $(\mathbf{H}^{m}_{s,*})^{-1} \cdot \big(\textrm{span}(\mathbf{H}^{m}_{s,p}\mathbf{V}^{m}_{p}) \cap \textrm{span}(\mathbf{H}^{m}_{s,q}\mathbf{V}^{m}_{q})\big)$, where $\mathbf{H}^{m}_{s,*}$ is an $m \times m$ diagonal matrix whose changing pattern of diagonal entries is the same as $\mathbf{H}^{m}_{s,p}$ and $\mathbf{H}^{m}_{s,q}$. By Lemma 1,
\begin{eqnarray*}
\mathbf{H}^{m}_{s,*}\mathbf{v}^{m}_{*} \in \textrm{span}(\mathbf{H}^{m}_{s,p}\mathbf{V}^{m}_{p})
\,\,\Rightarrow\,\, \mathbf{v}^{m}_{*} \in \textrm{span}(\mathbf{V}^{m}_{p}),
\end{eqnarray*}
since $\mathbf{H}^{m}_{s,*}$ and $\mathbf{H}^{m}_{s,p}$ have the same changing pattern of their diagonal entries.
In the same manner,
\begin{eqnarray*}
\mathbf{v}^{m}_{*} \in \textrm{span}(\mathbf{V}^{m}_{q}).
\end{eqnarray*}
Thus, we can say that
\begin{eqnarray*}
\mathbf{v}^{m}_{*} \in \textrm{span}(\mathbf{V}^{m}_{p}) \cap \textrm{span}(\mathbf{V}^{m}_{q}),
\end{eqnarray*}
which implies
\begin{align}\label{eq:subset}
(\mathbf{H}^{m}_{s,*})^{-1} \cdot \big(\textrm{span}(\mathbf{H}^{m}_{s,p}\mathbf{V}^{m}_{p}) \cap \textrm{span}(\mathbf{H}^{m}_{s,q}\mathbf{V}^{m}_{q})\big)
\subseteq \textrm{span}(\mathbf{V}^{m}_{p}) \cap \textrm{span}(\mathbf{V}^{m}_{q}).
\end{align}

On the other hand, suppose that $\mathbf{v}^{m}_{*}$ is included in the vector space, $\textrm{span}(\mathbf{V}^{m}_{p}) \cap \textrm{span}(\mathbf{V}^{m}_{q}) \cap \textrm{span}(\mathbf{V}^{m}_{s})$.
If $\mathbf{H}^{m}_{s,*}\mathbf{v}^{m}_{*}$ is included in $\textrm{span}(\mathbf{H}^{m}_{s,q}\mathbf{V}^{m}_{q})$,
\begin{eqnarray}
\mathbf{H}^{m}_{s,*}\mathbf{v}^{m}_{*} \in \textrm{span}(\mathbf{H}^{m}_{s,q}\mathbf{V}^{m}_{q}) 
&\Rightarrow& \mathbf{H}^{m}_{s,*}\mathbf{v}^{m}_{*} \in \textrm{span}(\mathbf{H}^{m}_{s,s}(\mathbf{H}^{m}_{s,s})^{-1}\mathbf{H}^{m}_{s,q}\mathbf{V}^{m}_{q}), \nonumber \\ 
&\Rightarrow& \mathbf{v}^{m}_{*} \in \textrm{span}((\mathbf{H}^{m}_{s,s})^{-1}\mathbf{H}^{m}_{s,q}\mathbf{V}^{m}_{q}).
\end{eqnarray}
In this case, $\textrm{span}(\mathbf{V}^{m}_{s})$ and $\textrm{span}( (\mathbf{H}^{m}_{s,s})^{-1}\mathbf{H}^{m}_{s,q}\mathbf{V}^{m}_{q})$ have non-zero intersection since $\mathbf{v}^{m}_{*}$ is included in both of the two subspaces. Thus, it can be expressed as
\begin{eqnarray}
\textrm{dim}(\mathbf{H}^{m}_{s,s}\mathbf{V}^{m}_{s} \cap \mathbf{H}^{m}_{s,q}\mathbf{V}^{m}_{q}) > 0.
\end{eqnarray}
At $\text{R}_{s}$, the desired signal subspace is contaminated by the interfering signal subspace from $\text{T}_{q}$. It can be said that $\text{R}_{s}$ cannot decode its desired symbols, hence the vector $\mathbf{H}^{m}_{s,*}\mathbf{v}^{m}_{*}$ should not be included in $\textrm{span}(\mathbf{H}^{m}_{s,q}\mathbf{V}^{m}_{q})$ if $\mathbf{v}^{m}_{*}$ is included in $\textrm{span}(\mathbf{V}^{m}_{p}) \cap \textrm{span}(\mathbf{V}^{m}_{q}) \cap \textrm{span}(\mathbf{V}^{m}_{s})$. It can be also applied to $\textrm{span}(\mathbf{H}^{m}_{s,p}\mathbf{V}^{m}_{p})$, thus leading to
\begin{align}
(\mathbf{H}^{m}_{s,*})^{-1} \cdot \big( \textrm{span}(\mathbf{H}^{m}_{s,p}\mathbf{V}^{m}_{p}) \cap \textrm{span}(\mathbf{H}^{m}_{s,q}\mathbf{V}^{m}_{q})\big)
 \cap \big(\textrm{span}(\mathbf{V}^{m}_{p}) \cap \textrm{span}(\mathbf{V}^{m}_{q}) \cap \textrm{span}(\mathbf{V}^{m}_{s})\big) = \emptyset.\label{eq:joint}
\end{align}
From (\ref{eq:subset}) and (\ref{eq:joint}),
\begin{align*}
\textrm{dim}(\mathbf{V}^{m}_{p} \cap \mathbf{V}^{m}_{q}) 
&\geq \textrm{dim}\Big( (\mathbf{H}^{m}_{s,*})^{-1} \cdot \big(\textrm{span}(\mathbf{H}^{m}_{s,p}\mathbf{V}^{m}_{p}) \cap \textrm{span}(\mathbf{H}^{m}_{s,q}\mathbf{V}^{m}_{q})\big) \Big) + \textrm{dim}(\mathbf{V}^{m}_{p} \cap \mathbf{V}^{m}_{q} \cap \mathbf{V}^{m}_{s}) \\ \nonumber 
&= \textrm{dim}(\mathbf{H}^{m}_{s,p}\mathbf{V}^{m}_{p} \cap \mathbf{H}^{m}_{s,q}\mathbf{V}^{m}_{q}) 
 + \textrm{dim}(\mathbf{V}^{m}_{p} \cap \mathbf{V}^{m}_{q} \cap \mathbf{V}^{m}_{s}).
\end{align*}
Thus, from (\ref{eq:HVHV}), we have
\begin{align}
\label{eq:simple}
\textrm{dim}(\mathbf{V}^{m}_{p} \cap \mathbf{V}^{m}_{q}) - \textrm{dim}(\mathbf{V}^{m}_{p} \cap \mathbf{V}^{m}_{q} \cap \mathbf{V}^{m}_{s}) \geq (3\lambda -1)m.
\end{align}

This allows us to present the following lemma, which provides a bound on the intersection of the beamforming vector spaces of any two users which are in the same alignment set.

\begin{lem}\label{lem3}
For two vertices $p$ and $r$ in an alignment set, if $\Delta$ alignment edges are needed to be traversed to go from $p$ to $r$, then
\begin{eqnarray}
\label{eq:Vpr}
\textrm{dim}(\mathbf{V}^{m}_{p} \cap \mathbf{V}^{m}_{r}) \geq ((2\Delta +1)\lambda - \Delta)m.
\end{eqnarray}
\end{lem}
\begin{proof}
We prove the lemma by induction on $\Delta$. If $\Delta=1$, then inequality (\ref{eq:Vpr}) holds according to (\ref{eq:simple}). We now show that if inequality (\ref{eq:Vpr}) holds for two vertices connected with $\Delta$ alignment edges, it also holds for two vertices connected with $\Delta +1$ alignment edges. Suppose that vectices $p$ and $r_{1}$, $r_{1}$ and $r_{2}$, and $p$ and $r_{2}$ are connected with $\Delta$, $1$, and $\Delta +1$ alignment edges, respectively. 
By the property that if the vector space $\mathcal{A}$ is included in the vector space $\mathcal{B}$, i.e., $\mathcal{A} \subseteq \mathcal{B}$, then $\textrm{dim}(\mathcal{A}) \leq \textrm{dim}(\mathcal{B})$,
\begin{align}
\textrm{dim}(\mathbf{V}^{m}_{p} \cap \mathbf{V}^{m}_{r_{1}}) + \textrm{dim}(\mathbf{V}^{m}_{r_{1}} \cap \mathbf{V}^{m}_{r_{2}})  - \textrm{dim}(\mathbf{V}^{m}_{p} \cap \mathbf{V}^{m}_{r_{2}} \cap \mathbf{V}^{m}_{r_{1}}) &\leq \textrm{rank}(\mathbf{V}^{m}_{r_{1}}) \nonumber \\ 
&= \lambda m.
\end{align}
Thus, we can say that
\begin{align}
\textrm{dim}(\mathbf{V}^{m}_{p} \cap \mathbf{V}^{m}_{r_{2}}) &\geq
\textrm{dim}(\mathbf{V}^{m}_{p} \cap \mathbf{V}^{m}_{r_{2}} \cap \mathbf{V}^{m}_{r_{1}}) \nonumber \\ 
&\geq \textrm{dim}(\mathbf{V}^{m}_{p} \cap \mathbf{V}^{m}_{r_{1}}) + \textrm{dim}(\mathbf{V}^{m}_{r_{1}} \cap \mathbf{V}^{m}_{r_{2}}) - \lambda m \nonumber \\
&\geq ((2\Delta +1)\lambda - \Delta)m + (3\lambda -1)m - \lambda m \nonumber \\ 
&= ((2(\Delta +1) +1)\lambda - (\Delta +1))m. \label{eq:Vpr2r1}
\end{align}
This completes the proof.
\end{proof}


Now, assume that vertex $k$ has two incoming internal conflicts from vertices $i$ and $j$ with conflict distances $\Delta_{1}$ and $\Delta_{2}$, respectively. Without loss of generality, suppose $\Delta_{\textrm{min}}=\Delta_{1} \leq \Delta_{2}$, hence $\Delta_{2}$ must be $\Delta_{\textrm{min}}$ or $\Delta_{\textrm{min}}+1$ since vertices $i$ and $j$ are connected with an alignment edge. 
Now, since vertices $k$ and $i$ are connected with $\Delta_{\textrm{min}}$ alignment edges, Lemma \ref{lem3} implies that
\begin{eqnarray}
\label{eq:Vki}
\textrm{dim}(\mathbf{V}^{m}_{k} \cap \mathbf{V}^{m}_{i}) \geq ((2\Delta_{\textrm{min}} +1)\lambda - \Delta_{\textrm{min}})m.
\end{eqnarray}
Equation (\ref{eq:simple}) holds for vertices $i$, $j$, and $k$, thus
\begin{align}
\label{eq:Vijk}
\textrm{dim}(\mathbf{V}^{m}_{i} \cap \mathbf{V}^{m}_{j}) - \textrm{dim}(\mathbf{V}^{m}_{i} \cap \mathbf{V}^{m}_{j} \cap \mathbf{V}^{m}_{k}) \geq (3\lambda -1)m.
\end{align}
From (\ref{eq:Vki}) and (\ref{eq:Vijk}), we know that
\begin{eqnarray}
((2\Delta_{\textrm{min}} +1)\lambda - \Delta_{\textrm{min}})m + (3\lambda -1)m &\leq&
\textrm{dim}(\mathbf{V}^{m}_{k} \cap \mathbf{V}^{m}_{i}) + \textrm{dim}(\mathbf{V}^{m}_{i} \cap \mathbf{V}^{m}_{j}) \nonumber \\ 
&& - \textrm{dim}(\mathbf{V}^{m}_{i} \cap \mathbf{V}^{m}_{j} \cap \mathbf{V}^{m}_{k}) \nonumber \\ 
&\leq& \lambda m,
\end{eqnarray}
which implies that
\begin{eqnarray}
\lambda \leq \frac{\Delta_{\textrm{min}}+1}{2\Delta_{\textrm{min}}+3}.
\end{eqnarray}
This completes the proof.
\end{proof}

\begin{prop}\label{prop2}
For partially connected $K$-user interference networks with reconfigurable antenna switching among any number of preset modes at the receivers,
the linear symmetric DoF is upper-bounded by
\begin{eqnarray}
\emph{LDoF}_{\emph{sym}} \leq \frac{2L_{\emph{min,odd}}}{5L_{\emph{min,odd}}+1}.
\end{eqnarray}
\end{prop}
\begin{proof}
Consider any internal conflict cycle of odd length $L$ in the network. Without loss of generality, assume this cycle is between nodes $[L+1,L,...,2]$, all receiving interference from transmitter $\text{T}_1$. This implies that each node $i \in [2:L]$ has incoming internal conflicts from nodes 1 and $i+1$, and node $L+1$ has incoming internal conflicts from nodes 1 and 2. Let us remove all the other interfering links in the network topology, which does not hurt the symmetric DoF. As in the proof of Proposition \ref{prop1}, suppose that each user sends $\lambda m$ data symbols over $m$ channel uses; i.e., $n_i = \lambda m, \forall i \in [1:K]$. Equation (\ref{eq:simple}) holds for any set of three vertices where two of them have outgoing conflict edges towards the third one. Thus,
\begin{align}
\label{eq:simple2}
\dim(\mathbf{V}_p^m \cap \mathbf{V}_q^m) - \dim(\mathbf{V}_p^m \cap \mathbf{V}_q^m \cap \mathbf{V}_s^m)\geq (3\lambda-1)m, ~ \forall \{p,q\}=\mathcal{I}_s.
\end{align}
For $s\in\{3,4\}$, we can write the above inequality as
\begin{align}
\dim(\mathbf{V}_1^m \cap \mathbf{V}_4^m) - \dim(\mathbf{V}_1^m \cap \mathbf{V}_4^m \cap \mathbf{V}_3^m) &\geq (3\lambda-1)m \label{eq:rx3}\\
\dim(\mathbf{V}_1^m \cap \mathbf{V}_5^m) - \dim(\mathbf{V}_1^m \cap \mathbf{V}_5^m \cap \mathbf{V}_4^m) &\geq (3\lambda-1)m. \label{eq:rx4}
\end{align}
On the other hand, we have
\begin{align}
\lambda m=\text{rank}(\mathbf{V}_1^m) &\geq \dim(\mathbf{V}_1^m \cap \mathbf{V}_3^m)+\dim(\mathbf{V}_1^m \cap \mathbf{V}_4^m)+\dim(\mathbf{V}_1^m \cap \mathbf{V}_5^m) \nonumber \\
& ~~ - \dim(\mathbf{V}_1^m \cap \mathbf{V}_4^m \cap \mathbf{V}_3^m) - \dim(\mathbf{V}_1^m \cap \mathbf{V}_5^m \cap \mathbf{V}_4^m)- \dim(\mathbf{V}_1^m \cap \mathbf{V}_3^m \cap \mathbf{V}_5^m).\label{eq:union}
\end{align}
Combining \eqref{eq:rx3}-\eqref{eq:rx4} with \eqref{eq:union}, we will get
\begin{align*}
2(3\lambda-1)m + \dim(\mathbf{V}_1^m \cap \mathbf{V}_3^m ) - \dim(\mathbf{V}_1^m \cap \mathbf{V}_3^m \cap \mathbf{V}_5^m) \leq \lambda m.
\end{align*}
This implies that in general, for $s\in\{i,i+1\}, \forall i\in \{3,5,7,...,L\}$, we can write
\begin{align*}
2(3\lambda-1)m + \dim(\mathbf{V}_1^m \cap \mathbf{V}_i^m ) - \dim(\mathbf{V}_1^m \cap \mathbf{V}_i^m \cap \mathbf{V}_{i+2}^m) \leq \lambda m,
\end{align*}
assuming $\mathbf{V}_{L+2}^m=\mathbf{V}_{2}^m$. Adding all these inequalities together, we will get
\begin{align}
\label{eq:sumV}
(L-1)(3\lambda-1)m + \sum_{\substack{i=3 \\ i \text{ odd}}}^L \left[\dim(\mathbf{V}_1^m \cap \mathbf{V}_i^m ) - \dim(\mathbf{V}_1^m \cap \mathbf{V}_i^m \cap \mathbf{V}_{i+2}^m) \right] \leq \frac{L-1}{2}\lambda m.
\end{align}
On the other hand, for $s=2$, we can write \eqref{eq:simple2} as
\begin{align}
(3\lambda-1)m \leq \dim(\mathbf{V}_1^m \cap \mathbf{V}_3^m) - \dim(\mathbf{V}_1^m \cap \mathbf{V}_3^m \cap \mathbf{V}_2^m). \label{eq:rx2}
\end{align}
Adding \eqref{eq:sumV} and \eqref{eq:rx2} and rearranging the terms, we will get
\begin{align}
\label{eq:subfin}
 & \dim(\mathbf{V}_1^m \cap \mathbf{V}_3^m \cap \mathbf{V}_2^m) - \dim(\mathbf{V}_1^m \cap \mathbf{V}_3^m \cap \mathbf{V}_5^m)  \nonumber \\
&\quad+ \sum_{\substack{i=5 \\ i \text{ odd}}}^L \left[\dim(\mathbf{V}_1^m \cap \mathbf{V}_i^m ) - \dim(\mathbf{V}_1^m \cap \mathbf{V}_i^m \cap \mathbf{V}_{i+2}^m) \right] \leq \left[ L-\frac{5L+1}{2} \lambda \right] m.
\end{align}

Clearly, for $L=3$, the LHS of \eqref{eq:subfin} is equal to zero since $\mathbf{V}_5^m=\mathbf{V}_{L+2}^m=\mathbf{V}_2^m$. We will now present the following lemma which states that for odd values of $L\geq 5$, the LHS of \eqref{eq:subfin} is non-negative.

\begin{lem}\label{lem4}
The following inequality holds for odd $L\geq 5$.
\begin{align}\label{eq:lem_L}
\sum_{\substack{i=5 \\ i \emph{ odd}}}^{L-2}\left[\emph{dim}(\mathbf{V}^{m}_{1} \cap \mathbf{V}^{m}_{i})-\emph{dim}(\mathbf{V}^{m}_{1} \cap \mathbf{V}^{m}_{i} \cap \mathbf{V}^{m}_{i+2})\right] + \emph{dim}(\mathbf{V}^{m}_{1} \cap \mathbf{V}^{m}_{L}) \geq \emph{dim}(\mathbf{V}^{m}_{1} \cap \mathbf{V}^{m}_{5,7,\ldots,L}),
\end{align}
where $\mathbf{V}^{m}_{5,7,\ldots,L}$ denotes $\emph{span}([\mathbf{V}^{m}_{5} ~ \mathbf{V}^{m}_{7} ~ \cdots ~ \mathbf{V}^{m}_{L}])$.
\end{lem}

\begin{proof}
We prove the lemma by induction on $L$. If $L=5$, then (\ref{eq:lem_L}) holds with equality, since both sides will be equal to $\dim(\mathbf{V}^{m}_{1} \cap \mathbf{V}^{m}_{5})$.
We now show that if inequality (\ref{eq:lem_L}) holds for any odd $L\geq 5$, it also holds for $L+2$. If inequality (\ref{eq:lem_L}) holds for $L$, then we can write
\begin{align*}
&\sum_{\substack{i=5 \\ i \textrm{ odd}}}^{L}\left[\textrm{dim}(\mathbf{V}^{m}_{1} \cap \mathbf{V}^{m}_{i})-\textrm{dim}(\mathbf{V}^{m}_{1} \cap \mathbf{V}^{m}_{i} \cap \mathbf{V}^{m}_{i+2})\right] + \textrm{dim}(\mathbf{V}^{m}_{1} \cap \mathbf{V}^{m}_{L+2})  \\ 
& \qquad \geq \textrm{dim}(\mathbf{V}^{m}_{1} \cap \mathbf{V}^{m}_{5,7,\ldots,L}) - \textrm{dim}(\mathbf{V}^{m}_{1} \cap \mathbf{V}^{m}_{L} \cap \mathbf{V}^{m}_{L+2}) + \textrm{dim}(\mathbf{V}^{m}_{1} \cap \mathbf{V}^{m}_{L+2})  \\ 
& \qquad \geq \textrm{dim}(\mathbf{V}^{m}_{1} \cap \mathbf{V}^{m}_{5,7,\ldots,L}) - \textrm{dim}(\mathbf{V}^{m}_{1} \cap \mathbf{V}^{m}_{5,7,\ldots,L} \cap \mathbf{V}^{m}_{L+2}) + \textrm{dim}(\mathbf{V}^{m}_{1} \cap \mathbf{V}^{m}_{L+2}) \\ 
& \qquad = \textrm{dim}(\mathbf{V}^{m}_{1} \cap \mathbf{V}^{m}_{5,7,\ldots,L+2}).
\end{align*}
Therefore, the inequality is also true for $L+2$ and the proof is complete.
\end{proof}


Now, we can write the LHS of \eqref{eq:subfin} as
\begin{align*}
& \Scale[.95]{\dim(\mathbf{V}_1^m \cap \mathbf{V}_3^m \cap \mathbf{V}_2^m) - \dim(\mathbf{V}_1^m \cap \mathbf{V}_3^m \cap \mathbf{V}_5^m)  +} \sum_{\substack{i=5 \\ i \text{ odd}}}^L \Scale[.95]{ \left[\dim(\mathbf{V}_1^m \cap \mathbf{V}_i^m ) - \dim(\mathbf{V}_1^m \cap \mathbf{V}_i^m \cap \mathbf{V}_{i+2}^m) \right]}\\
& ~  \Scale[.95]{=\textrm{dim}(\mathbf{V}^{m}_{1} \cap \mathbf{V}^{m}_{3} \cap \mathbf{V}^{m}_{2}) - \textrm{dim}(\mathbf{V}^{m}_{1} \cap \mathbf{V}^{m}_{3} \cap \mathbf{V}^{m}_{5}) +} \sum_{\substack{i=5 \\ i \text{ odd}}}^{L-2} \Scale[.95]{\left[\textrm{dim}(\mathbf{V}^{m}_{1} \cap \mathbf{V}^{m}_{i})-\textrm{dim}(\mathbf{V}^{m}_{1} \cap \mathbf{V}^{m}_{i} \cap \mathbf{V}^{m}_{i+2})\right]} \\ 
&  \quad \Scale[.95]{+ \textrm{dim}(\mathbf{V}^{m}_{1} \cap \mathbf{V}^{m}_{L})-\textrm{dim}(\mathbf{V}^{m}_{1} \cap \mathbf{V}^{m}_{L} \cap \mathbf{V}^{m}_{2})} \\ 
& ~ \Scale[.95]{\overset{(a)}{\geq} \textrm{dim}(\mathbf{V}^{m}_{1} \cap \mathbf{V}^{m}_{3} \cap \mathbf{V}^{m}_{2}) - \textrm{dim}(\mathbf{V}^{m}_{1} \cap \mathbf{V}^{m}_{3} \cap \mathbf{V}^{m}_{5}) + \textrm{dim}(\mathbf{V}^{m}_{1} \cap \mathbf{V}^{m}_{5,7,\ldots,L})-\textrm{dim}(\mathbf{V}^{m}_{1} \cap \mathbf{V}^{m}_{L} \cap \mathbf{V}^{m}_{2})} \\ 
&~ \Scale[.95]{\geq \textrm{dim}(\mathbf{V}^{m}_{1} \cap \mathbf{V}^{m}_{3} \cap \mathbf{V}^{m}_{2}) - \textrm{dim}(\mathbf{V}^{m}_{1} \cap \mathbf{V}^{m}_{3} \cap \mathbf{V}^{m}_{5,7,\ldots,L}) + \textrm{dim}(\mathbf{V}^{m}_{1} \cap \mathbf{V}^{m}_{5,7,\ldots,L})} \\ 
& \quad \Scale[.95]{-\textrm{dim}(\mathbf{V}^{m}_{1} \cap \mathbf{V}^{m}_{5,7,\ldots,L} \cap \mathbf{V}^{m}_{2})} \\ 
&~ \Scale[.95]{\geq \textrm{dim}(\mathbf{V}^{m}_{1} \cap \mathbf{V}^{m}_{3} \cap \mathbf{V}^{m}_{5,7,\ldots,L} \cap \mathbf{V}^{m}_{2})  - \textrm{dim}(\mathbf{V}^{m}_{1} \cap \mathbf{V}^{m}_{3} \cap \mathbf{V}^{m}_{5,7,\ldots,L}) + \textrm{dim}(\mathbf{V}^{m}_{1} \cap \mathbf{V}^{m}_{5,7,\ldots,L})} \\ 
& \quad \Scale[.95]{-\textrm{dim}(\mathbf{V}^{m}_{1} \cap \mathbf{V}^{m}_{5,7,\ldots,L} \cap \mathbf{V}^{m}_{2})}  \\ 
&~ \Scale[.95]{\geq 0,}
\end{align*}
where in (a), we have invoked Lemma \ref{lem4}. This, together with \eqref{eq:subfin}, implies that 
\begin{eqnarray*}
\lambda \leq \frac{2L}{5L+1}.
\end{eqnarray*}
Since the function $f(L)=\frac{2L}{5L+1}$ is decreasing in $L$, considering the shortest cycle of odd length $L_{\text{min,odd}}$ yields the tightest upper bound, hence completing the proof.
\end{proof}

\section{Achievable Scheme for Topologies with at most Two Co-interferers}
\label{sec_2_coint}
In this section, we prove Theorem \ref{thm3} by showing that the upper bound on the linear symmetric DoF in Theorem \ref{thm2} is tight if all the transmitters have at most two co-interferers.
We also introduce two examples: one is a network topology that can achieve the upper bound, and the other is a network topology in which there exists a fork on the alignment graph and cannot achieve the upper bound. 

\begin{proof}[Proof of Theorem 3]
For a network topology in which the transmitters have at most two co-interferers, we show that 
each transmitter can send $\Delta_{\textrm{min}}+1$ data symbols over $2\Delta_{\textrm{min}}+3$ channel uses; i.e., $n_{i}=\Delta_{\textrm{min}}+1, \forall i \in [1:K]$ and $m=2\Delta_{\textrm{min}}+3$. Since in such topologies, the alignment graph has no forks, each alignment set can be represented either as a concatenated line or a cycle of alignment edges. The beamforming vectors are determined to meet the following conditions so that the interfering signals are aligned into $(\Delta_{\textrm{min}}+2)$-dimensional signal subspace.

\begin{itemize}
\item Each vertex has $\Delta_{\textrm{min}}+1$ beamforming vectors.
\item Two vertices connected with an alignment edge have at least $\Delta_{\textrm{min}}$ common beamforming vectors that the vertices to which $\Delta_{\textrm{min}}$ or more alignment edges are needed to be traversed to go from both of them do not have.
\end{itemize}

The above conditions can be realized when vertices have the beamforming vectors in a consecutive order by removing the first one and adding a new one at the last (for a concatenated line) or sliding out each other's way as meeting start and end (for a cycle). Two examples that show how to determine the beamforming vectors with $\Delta_{\textrm{min}}=1$ are given in Fig \ref{exBV}. 
Moreover, to guarantee that the desired signal subspaces are independent of the interference signal subspaces, the following conditions should be satisfied.

\begin{itemize}
\item If the source and the destination of a conflict edge have common beamforming vectors, these vectors should have at least two non-zero entries. 
\item Two vertices that are not included in the same alignment set should not have common beamforming vectors.
\item Any $2\Delta_{\textrm{min}}+3$ beamforming vectors should be linearly independent.
\end{itemize}
\begin{figure}[t]
    \centerline{\includegraphics[width=8.0cm]{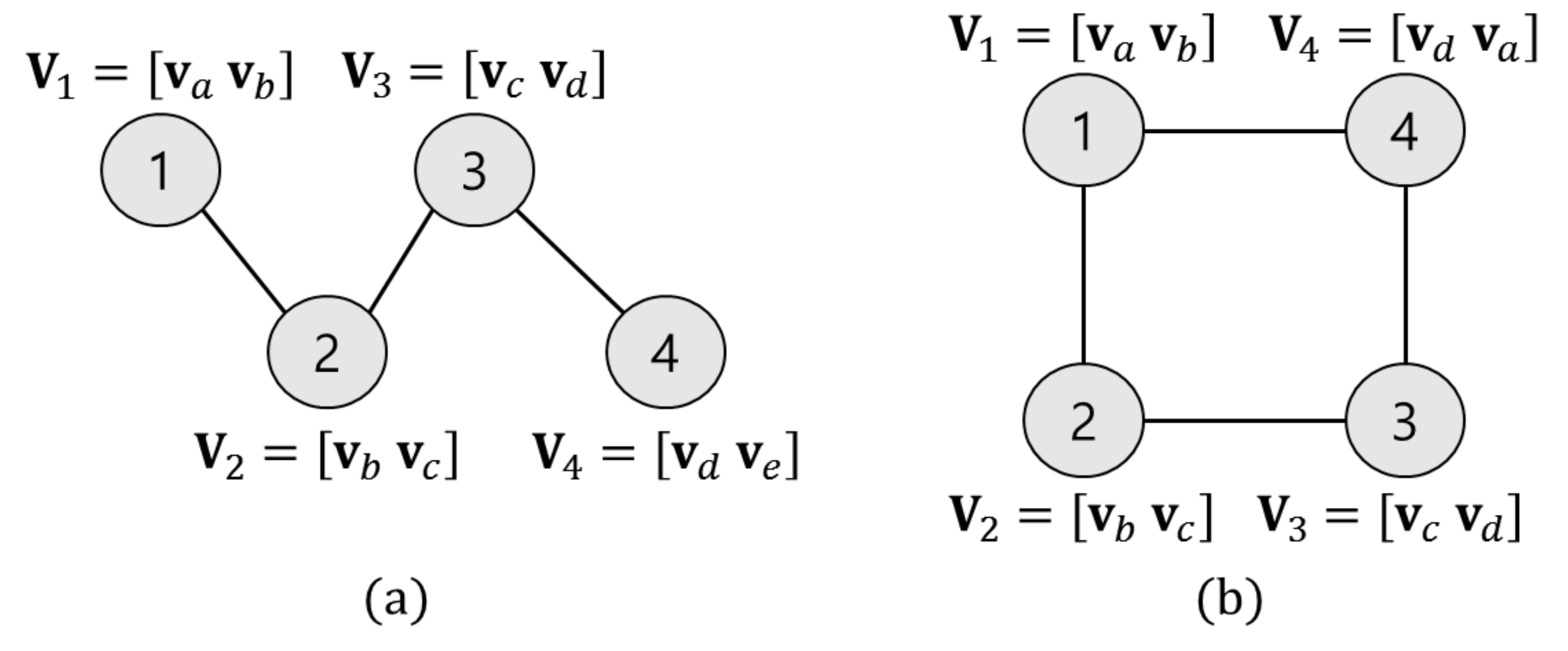}}
    \caption{Two examples of the beamforming vectors for 4 vertices of an alignment set with $\Delta_{\textrm{min}}=1$. ((a) a concatanated line, and (b) a cycle)}
    \label{exBV}
    \vspace{-4mm}
\end{figure}

Note that there cannot be more than three interferers at any receiver with the constraint of at most two co-interferers. If $\text{R}_{k}$ is interfered by three transmitters, interferers do not have any other co-interferer since the number of their co-interferers is already two. Thus, their corresponding vertices are connected with triangle-shaped alignment edges and $k$ does not belong to this alignment set. Since they have a total of $\Delta_{\textrm{min}}+2$ beamforming vectors independent of the beamforming vectors of $\text{T}_{k}$, $\text{R}_{k}$ can achieve $\frac{\Delta_{\textrm{min}}+1}{2\Delta_{\textrm{min}}+3}$ linear DoF.

For the receivers interfered by less than three transmitters, we can classify them into three cases by the number of incoming internal conflicts.

\emph{1) Two incoming internal conflicts}:
In this case, the receiver is connected to two transmitters in the same alignment set. We assume that vertex $k$ has two incoming internal conflicts from $i$ and $j$ with minimum conflict distance $\Delta_{\textrm{min}}$. As mentioned above, $\text{T}_{i}$ and $\text{T}_{j}$ have common $\Delta_{\textrm{min}}$ beamforming vectors that $\text{T}_{k}$ does not have. It can be denoted as
\begin{eqnarray}
\mathbf{v}^{m}_{i,n} &=& \mathbf{v}^{m}_{j,n}, \,\,\,\,\, \forall n \in [2:\Delta_{\textrm{min}}+1].
\end{eqnarray}
For the beamforming vectors $\mathbf{v}^{m}_{i,n}$ and $\mathbf{v}^{m}_{j,n}, \forall n \in [2:\Delta_{\textrm{min}}+1]$, the interfering signals $\mathbf{H}^{m}_{k,i}\mathbf{v}^{m}_{i,n}$ and $\mathbf{H}^{m}_{k,j}\mathbf{v}^{m}_{j,n}$ can be aligned into a 1-dimensional signal subspace if $\text{R}_{k}$ has the identical preset mode at the time slots when $\mathbf{v}^{m}_{i,n}$ and $\mathbf{v}^{m}_{j,n}$ have non-zero entries at the corresponding rows. On the other hand, the desired signal $\mathbf{H}^{m}_{k,k}\mathbf{v}^{m}_{k,n}, \forall n \in [1:\Delta_{\textrm{min}}+1]$ can have an independent signal subspace at $\text{R}_{k}$ although $\mathbf{v}^{m}_{k,n}$ is a common beamforming vector with $\text{T}_{i}$ or $\text{T}_{j}$, if $\text{R}_{k}$ has at least two preset modes at the time slots when $\mathbf{v}^{m}_{k,n}$ has non-zero entries at the corresponding rows.
The received signal at $\text{R}_{k}$ over $2\Delta_{\textrm{min}}+3$ channel uses is given by
\begin{align*}
 \mathbf{y}^{m}_{k} =  \mathbf{H}^{m}_{k,k}\mathbf{V}^{m}_{k}\mathbf{s}^{m}_{k} + \underbrace{\sum_{n=2}^{\Delta_{\textrm{min}}+1} \left(\mathbf{H}^{m}_{k,i}\mathbf{v}^{m}_{i,n}s_{i,n}
+\mathbf{H}^{m}_{k,j}\mathbf{v}^{m}_{j,n}s_{j,n}\right)}_{\textrm{rank} =  \Delta_{\textrm{min}}}
 + \underbrace{\left(\mathbf{H}^{m}_{k,i}\mathbf{v}^{m}_{i,1}s_{i,1}
+\mathbf{H}^{m}_{k,j}\mathbf{v}^{m}_{j,1}s_{j,1}\right)}_{\textrm{rank} =  2} + \mathbf{z}^{m}_{k}.
\end{align*}

Hence, $\text{R}_{k}$ can decode the desired $\Delta_{\textrm{min}}+1$ symbols over $2\Delta_{\textrm{min}}+3$ channel uses, thus $\frac{\Delta_{\textrm{min}}+1}{2\Delta_{\textrm{min}}+3}$ linear DoF is achievable.

\emph{2) One incoming internal conflict}:
If vertex $k$ has a single incoming internal conflict from vertex $i$, $\text{R}_{k}$ is only interfered by $\text{T}_{i}$. It can decode all desired symbols since it receives only $2(\Delta_{\textrm{min}}+1)$ symbols over $2\Delta_{\textrm{min}}+3$ channel uses, although $\text{T}_{k}$ and $\text{T}_{i}$ may have some common beamforming vectors. If $\mathbf{V}^{m}_{k}$ and $\mathbf{V}^{m}_{i}$ have the same columns, $\text{R}_{k}$ should have at least two preset modes at the time slots when they have non-zero entries at the corresponding rows.

\emph{3) No incoming internal conflicts}:
If vertex $k$ has no incoming internal conflicts, $\text{R}_{k}$ receives the interfering signals from the transmitters which are included in another alignment set $\mathcal{S}$. Since two vertices connected with an alignment edge have the same $\Delta_{\textrm{min}}$ beamforming vectors, the beamforming vectors of one or two interferers can be determined with at most $\Delta_{\textrm{min}}+2$ vectors. Thus, the interfering signals can be aligned into $(\Delta_{\textrm{min}}+2)$-dimensional signal subspace at $\text{R}_{k}$ without preset mode switching. The received signal at $\text{R}_{k}$ over $2\Delta_{\textrm{min}}+3$ channel uses is given by
\begin{eqnarray}
\mathbf{y}^{m}_{k} = \mathbf{H}^{m}_{k,k}\mathbf{V}^{m}_{k}\mathbf{s}_{k} + \underbrace{\sum\nolimits_{i \in \mathcal{I}_k} \mathbf{H}^{m}_{k,i}\mathbf{V}^{m}_{i}\mathbf{s}_{i}}_{\textrm{rank} \leq \Delta_{\textrm{min}}+2}+\mathbf{z}^{m}_{k}.
\end{eqnarray}
Therefore, the linear DoF of $\frac{\Delta_{\textrm{min}}+1}{2\Delta_{\textrm{min}}+3}$ is achievable in all cases and the proof is complete.
\end{proof}

\begin{ex}

\begin{figure}[t]
    \centerline{\includegraphics[width=8.0cm]{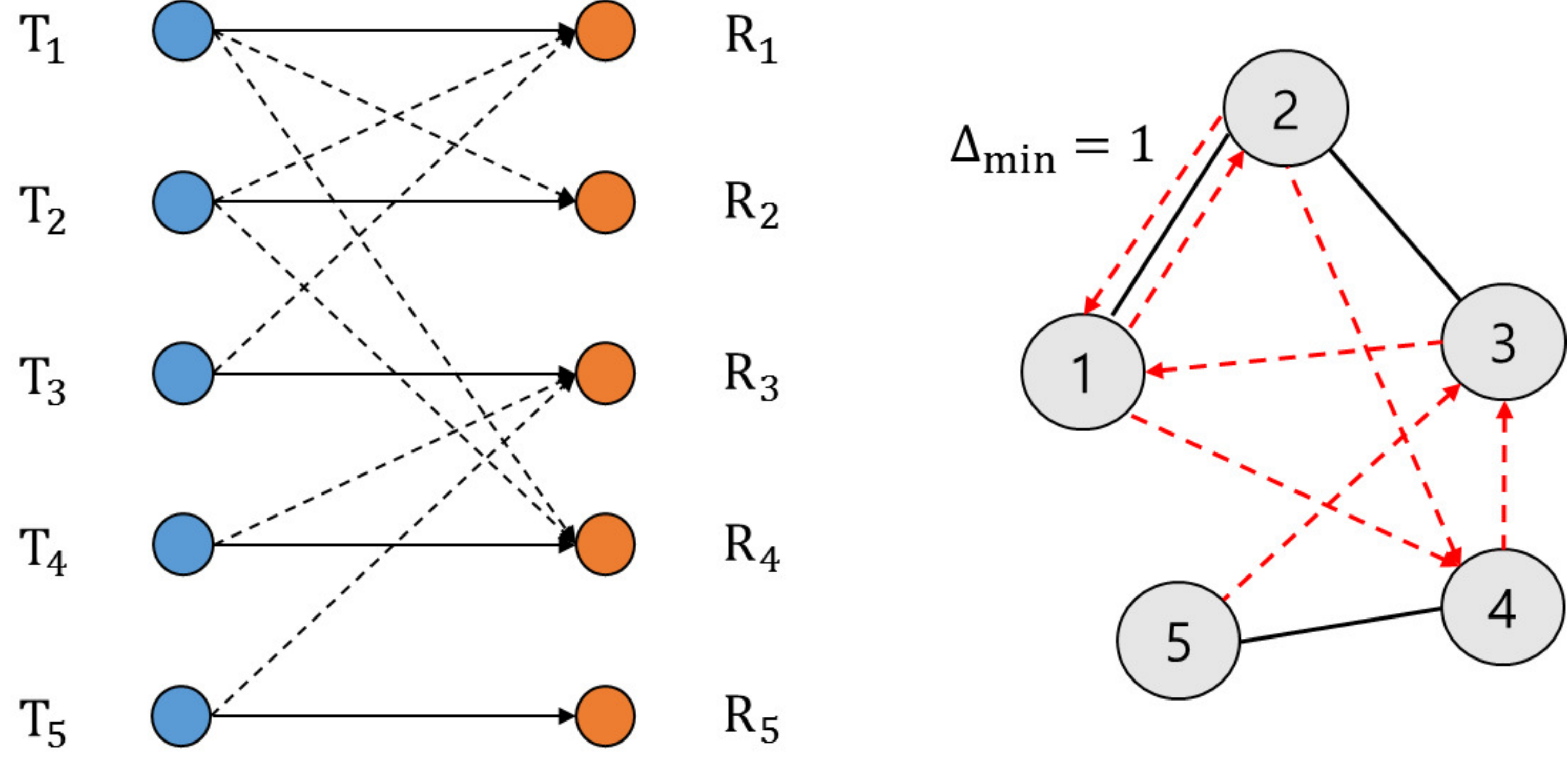}}
    \caption{A $5$-user interference network and its alignment and conflict graphs.}
    \label{DoF25}
    \vspace{-4mm}
\end{figure}

Consider a partially connected $5$-user interference network illustrated in Fig. \ref{DoF25}. According to its corresponding alignment and conflict graphs, the alignment sets are determined as
\begin{eqnarray}
\mathcal{S}_{1} = \{1, 2, 3\}, \,\,\,\,\, \mathcal{S}_{2} = \{4, 5\}.
\end{eqnarray}
In the alignment and conflict graphs, vertex $1$ has two incoming internal conflicts with $\Delta_{\textrm{min}}=1$. Therefore, since the number of co-interferers of all the transmitters is bounded by 2, Theorem 3 implies that a linear symmetric DoF of $\frac{2}{5}$ is achievable, through a  scheme where each transmitter sends 2 data symbols over 5 channel uses. According to the introduced linear scheme, the beamforming matrices are determined as
\begin{eqnarray*}
\mathbf{V}^{5}_{1}
= [\mathbf{v}^{5}_{a}  \hspace{3mm} \mathbf{v}^{5}_{b}], \,\,\,\,\,
\mathbf{V}^{5}_{2}
= [\mathbf{v}^{5}_{b} \hspace{3mm}  \mathbf{v}^{5}_{c}], \,\,\,\,\,
\mathbf{V}^{5}_{3}
= [\mathbf{v}^{5}_{c}  \hspace{3mm} \mathbf{v}^{5}_{d}], \,\,\,\,\,
\mathbf{V}^{5}_{4}
= [\mathbf{v}^{5}_{e} \hspace{3mm}  \mathbf{v}^{5}_{f}], \,\,\,\,\,
\mathbf{V}^{5}_{5}
= [\mathbf{v}^{5}_{e} \hspace{3mm}  \mathbf{v}^{5}_{f}].
\end{eqnarray*}
At $\text{R}_{1}$, the interfering signals can be aligned into a $3$-dimensional signal subspace since $\text{T}_{2}$ and $\text{T}_{3}$ have a common beamforming vector $\mathbf{v}^{5}_{c}$ which $\text{T}_{1}$ does not have. Although $\text{T}_{1}$ and $\text{T}_{2}$ have a common beamforming vector $\mathbf{v}^{5}_{b}$ and they interfere at $\text{R}_{2}$ and $\text{R}_{1}$, respectively, the received signals of $\mathbf{v}^{5}_{b}$ occupy an independent signal subspace at $\text{R}_{2}$ and $\text{R}_{1}$ by changing the preset modes at these receivers. In order to make the signal and interference subspaces independent at $\text{R}_{1}$ and $\text{R}_{2}$ by changing the preset modes, $\mathbf{v}^{5}_{b}$ should have at least two non-zero entries. In addition, any five of the beamforming vectors should be linearly independent. Under the above restrictions, we can pick the following beamforming vectors
\begin{eqnarray*}
&&\hspace{-5mm}\mathbf{v}^{5}_{a}=[1\hspace{2mm}0\hspace{2mm}0\hspace{2mm}0\hspace{2mm}0]^{T}, \,\,\,\,\,\,\,\,\,\, 
\mathbf{v}^{5}_{b}=[0\hspace{2mm}1\hspace{2mm}1\hspace{2mm}0\hspace{2mm}0]^{T}, \,\,\,\,\,\,\,\,\,\, 
\mathbf{v}^{5}_{c}=[0\hspace{2mm}0\hspace{2mm}0\hspace{2mm}1\hspace{2mm}0]^{T}, \\ \nonumber
&&\hspace{-5mm}\mathbf{v}^{5}_{d}=[0\hspace{2mm}0\hspace{2mm}0\hspace{2mm}0\hspace{2mm}1]^{T}, \,\,\,\,\,\,\,\,\,\, 
\mathbf{v}^{5}_{e}=[0\hspace{2mm}1\hspace{2mm}0\hspace{2mm}0\hspace{2mm}0]^{T},\,\,\,\,\,\,\,\,\,\, 
\mathbf{v}^{5}_{f}=[1\hspace{2mm}0\hspace{2mm}1\hspace{2mm}1\hspace{2mm}1]^{T}.
\end{eqnarray*}

Moreover, we select the preset modes of $\text{R}_{1}$ and $\text{R}_{2}$ as
\begin{eqnarray}
\mathbf{L}^{5}_{1}=\mathbf{L}^{5}_{2}= [1\hspace{2mm}1\hspace{2mm}2\hspace{2mm}1\hspace{2mm}1].
\end{eqnarray}
where the preset mode changes when the
common beamforming vector of an interferer and the corresponding transmitter, i.e., $\mathbf{v}^{5}_{b}$, has non-zero entries. Thus, $\text{R}_{1}$ and $\text{R}_{2}$ can decode $s_{1,2}$ and $s_{2,1}$ since they have two independent equations at time slots 2 and 3. They can decode $s_{1,1}$ and $s_{2,2}$ at time slots 1 and 4, respectively.
The other receivers $\text{R}_{3}$, $\text{R}_{4}$, and $\text{R}_{5}$ do not need to change their preset mode since vertices $3$, $4$, and $5$ have no internal conflicts, thus their corresponding transmitters have no common beamforming vectors with interferers. 
They can decode their desired symbols by appropriate signal subtractions.
Therefore, the linear symmetric DoF of $\frac{\Delta_{\min}+1}{2\Delta_{\min}+3}=\frac{2}{5}$ is achievable in this network.


\end{ex}

\begin{ex}\label{ex3}

\begin{figure}[t]
    \centerline{\includegraphics[width=8.0cm]{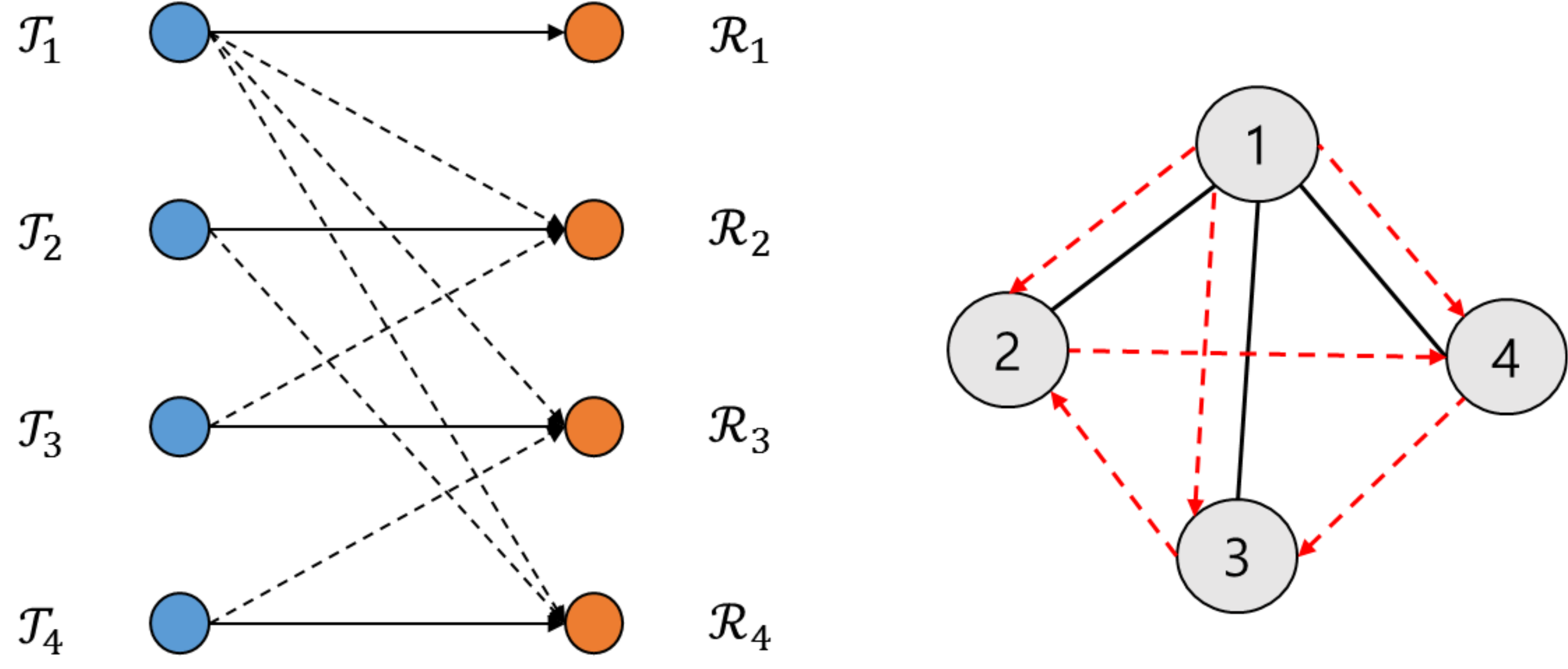}}
    \caption{A $4$-user network topology and its corresponding alignment and conflict graphs.}
    \label{DoF38}
    \vspace{-3mm}
\end{figure}

As the final example, consider the network topology in Fig. \ref{DoF38}. As the alignment and conflict graphs in Fig. \ref{DoF38} illustrate, vertex 1 is a fork since it is connected to vertices 2, 3, and 4 with alignment edges. According to Theorem 2, the linear symmetric DoF for this network topology is upper-bounded by $\textrm{min}\left(\frac{2}{5},\frac{3}{8}\right)=\frac{3}{8}$, since the minimum conflict distance $\Delta_{\textrm{min}}$ is equal to 1 and there exists an odd-length internal conflict cycle of length $L_{\text{min,odd}}=3$ among nodes $[4,3,2]$.

In this network topology, the upper bound of $\frac{3}{8}$ on the linear symmetric DoF can be achieved if the beamforming matrices are determined as
\begin{eqnarray}
\mathbf{V}^{8}_{1}=[\mathbf{v}^{5}_{a}\,\,\mathbf{v}^{5}_{b}\,\,\mathbf{v}^{5}_{c}], \,\,\,\,\,\,
\mathbf{V}^{8}_{2}=[\mathbf{v}^{5}_{a}\,\,\mathbf{v}^{5}_{d}\,\,\mathbf{v}^{5}_{e}], \,\,\,\,\,\,
\mathbf{V}^{8}_{3}=[\mathbf{v}^{5}_{b}\,\,\mathbf{v}^{5}_{f}\,\,\mathbf{v}^{5}_{g}], \,\,\,\,\,\,
\mathbf{V}^{8}_{4}=[\mathbf{v}^{5}_{c}\,\,\mathbf{v}^{5}_{h}\,\,\mathbf{v}^{5}_{i}].
\end{eqnarray}
Using these beamforming matrices, interfering signals are aligned into 5-dimensional subspaces and the desired signals occupy 3-dimensional signal subspaces independent of the interference signal subspaces. This implies that a linear symmetric DoF of $\frac{3}{8}$ is achievable, which meets the upper bound of Theorem \ref{thm2}.
\end{ex}

\section{Concluding Remarks}
\label{sec_con}
In this paper, we considered the problem of topological interference management with reconfigurable antennas at the receivers. We characterized the network topologies in which half linear symmetric DoF is achievable. We also derived the linear symmetric DoF upper bound in terms of the minimum conflict distance among internal conflicts towards the vertices that have at least two incoming internal conflicts and the minimum odd length of internal conflict cycles. This upper bound is shown to be achievable if there are no forks in the alignment graph of the network. However, if there is a fork in the alignment graph of a network topology, the linear symmetric DoF upper bound is not always tight. Thus, an interesting future direction would be to characterize the linear symmetric DoF for arbitrary network topologies with reconfigurable antennas.

Another interesting topic to study on the problem of topological interference management with reconfigurable antennas at the receivers is to characterize the sum-DoF of the network topologies, in contrast to the symmetric DoF metric that we considered in the paper. As an example, in a \emph{one-to-many} topology where a single transmitter interferes at all unintended receivers, turning that transmitter-receiver pair off would possibly maximize the sum-DoF of the network, whereas such a scheme may be suboptimal for the case when the desired performance metric is the symmetric DoF. We leave the problem of characterizing the sum-DoF of the network topologies with reconfigurable receiver antennas as a future work.

\section*{Acknowledgments}
\addcontentsline{toc}{section}{Acknowledgment}
This work is in part supported by JPL grant RTD RSA No. 1537733, NSF grants CAREER 1408639, NETS-1419632, EARS-1411244, ONR award N000141612189, Basic Science Research Program (NRF-2013R1A1A2008956 and NRF-2015R1A2A1A15052493) through NRF funded by MOE and MSIP, Technology Innovation Program (10051928) funded by MOTIE, Bio-Mimetic Robot Research Center funded by DAPA (UD130070ID), INMAC, and BK21-plus.


\end{document}